\documentclass[twoside,a4paper]{amsart}

\usepackage[english]{babel}
\usepackage{amssymb}
\usepackage{dsfont}
\usepackage{amsmath}
\def\bkR{\mathds{R}}

\def\bkZ{\mathds{Z}}
\def\mcal{\mathcal}

\newtheorem{theorem}{Theorem}[section]
\newtheorem{corollary}[theorem]{Corollary}
\newtheorem{lemma}[theorem]{Lemma}
\newtheorem{notation}[theorem]{Notation}
\newtheorem{construction}[theorem]{Construction}
\newtheorem{assumptions}[theorem]{Assumptions}
\newtheorem{assumption}[theorem]{Assumption}
\newtheorem{definition}[theorem]{Definition}
\newtheorem{proposition}[theorem]{Proposition}
\newtheorem{remark}[theorem]{Remark}

\newtheorem{example}[theorem]{Example}
\newtheorem{examples}[theorem]{Examples}
\newtheorem{acknowledgements}[theorem]{Acknowledgements}

\usepackage{color}

\usepackage{enumerate}
\usepackage[all]{xy}

\def\bkS{\mathds{S}}
\def\bkT{\mathds{T}}
\def\bkF{\mathds{F}}
\def\bkZ{\mathds{Z}}

\def\bkM{\mathds{M}}
\newcommand{\Set}{\mathsf{Set}}
\newcommand{\A}{\mcal{A}}
\newcommand{\D}{\mcal{D}}
\newcommand{\C}{\mcal{C}}
\renewcommand{\P}{\mcal{P}}
\newcommand{\Alg}{\mathbf{Alg}}
\newcommand{\Gra}{\mathbf{Gra}}
\newcommand{\Ord}{\mathbf{Ord}}

\def\id{{\mathit{id}}}
\def\card{{\mathit{card}}}
\def\Id{{\mathit{Id}}}
\def\inl{{\mathit{inl}}}

\def\colim{\mathop{\mathrm{colim}}\limits}
\newcommand{\obj}{\mathsf{obj}}

 \author{Ji\v{r}\'{\i} Ad\'{a}mek}
 \address{Technical University Braunschweig}
 \email{adamek@iti.cs.tu-bs.de}
\title{Colimits of Monads}
 \dedicatory{Dedicated to the seventieth birthday of Manuela Sobral.}
 \def\labelauthor{JA} %%%INSERT Initials
 
\begin{document}

\begin{abstract}
	The category of all monads over many-sorted sets (and over other "set-like" categories) is proved to have coequalizers and strong cointersections. And a general diagram has a colimit whenever all the monads involved preserve monomorphisms and have arbitrarily large joint pre-fixpoints. In contrast, coequalizers fail to exist e.g. for monads over the (presheaf) category of graphs.
	
	For more general categories we extend the results on coproducts of monads from \cite{\labelauthor:bib:ablm}. We call a monad separated if, when restricted to monomorphisms, its unit has a complement. We prove that every collection of separated monads with arbitrarily large joint pre-fixpoints has a coproduct. And a concrete formula for these coproducts is presented.

\end{abstract}

\maketitle

%------------- BEGINNING OF THE BODY OF THE ARTICLE -----------------%%

%% --- Each section begins with --- %%
\section{Introduction}\label{Introduction}
\setcounter{equation}{0}
	Whereas limits in the category Monad($\A$) of monads over a complete category $\A$ are easy, since the forgetful functor into the category $[\A,\A]$ of all endofunctors creates limits, colimits are more interesting. For example, a coproduct of two monads need not exist in Monad ($\A$) -- in fact, there are only four (trivial) types of monads over $\Set$ having a coproduct with every monad, as proved in \cite{\labelauthor:bib:ablm}, see Theorem \ref{\labelauthor:thm:T-44} below. In that paper a formula for coproducts of monads over $\Set$ was presented, and we extend it to coproducts of separated monads over general categories $\A$. Separatedness means that a complement of the unit of the monad exists if we restrict ourselves to the category $\A_m$ of objects and monomorphisms of $\A$. All consistent monads over $\Set$ are separated, see \cite{\labelauthor:bib:ablm}. For other base categories many interesting monads fail to be separated.
	
	Our main result is that in "set-like" categories, e.\,g., many-sorted sets, vector spaces or sets and partial functions, the category Monad ($\A$) has (a) all coequalizers and strong cointersections and (b) colimits of every diagram of monos-preserving monads with arbitrarily large joint pre-fixpoints. (An object $X$ is a pre-fixpoint of a monad $\bkS$ if $SX$ is a subobject of $X$.) That last condition is proved to be weaker than assuming that the monads are accessible. Moreover, arbitrarily large joint pre-fixpoints are sufficient for coproducts of 
	\begin{enumerate}[(1)]
		\item monos-preserving monads over set-like categories
		\item separated monads over rather general categories.
	\end{enumerate}
	And if $\A = \Set$, this condition is in case of coproducts of consistent monads also necessary (unless all but one of the monads are of the trivial type, see Theorem \ref{\labelauthor:thm:T-44} below). It is an open problem whether having arbitrarily large joint pre-fixpoints is sufficient for coproducts of general monads over "reasonably" general categories.
	
	Colimits of monads were studied by Kelly \cite{\labelauthor:bib:Ke} who proved, inter alia, that for locally presentable base categories $\A$ every diagram of accessible monads has a colimit in Monad ($\A$). Kelly also proved a formula for the colimit. In case of coproducts of consistent monads over $\Set$ a much simpler formula was presented in \cite{\labelauthor:bib:ablm}, inspired by the work of Ghani and Ustalu \cite{\labelauthor:bib:gu}: let $\bkS$ and $\bkT$ be consistent $\lambda$-accessible monads with unit complements $\bar{S}$ and $\bar{T}$, respectively. Then the coproduct monad is given by
	\begin{equation*}
		A \mapsto A+ \colim_{i < \lambda} X_i + \colim_{i<\lambda} Y_i
	\end{equation*}
	Here $X_i$ and $Y_i$ are the $\lambda$-chains formed by colimits on limit ordinals, whereas the isolated steps are defined by the following mutual recursion:
	\begin{equation*}
		X_{i+1} = \bar{S}(Y_i + A) \textnormal{ and } Y_{i+1} = \bar{T}(X_i + A)
	\end{equation*}
	
	We prove that, unsurprisingly, the same formula holds for coproducts of separated monads on general categories.
	
\begin{acknowledgements}
	The author is very grateful to Paul Levy for his comments that have	improved the presentation. The fact that cointersections exist in the category of monads over $\Set^S$ (see Remark \ref{\labelauthor:def:rmk-311} below) follows from an independent argument that Levy has presented in a personal communication.
\end{acknowledgements}

\section{The Category of Monads}
	In this section some basic properties of the category of monads and monad morphisms are collected.
	
	\begin{notation}
		\begin{enumerate}[(a)]
			\item Given a category $\A$ we write $[\A,\A]$ for the category of endofunctors on $\A$ and natural transformations between them. And Monad ($\A$) denotes the category of monads and monad morphisms. The obvious forgetful functor is denoted by $V:$ Monad ($\A$) $\to [\A,\A]$.
			\item We use $\ast$ to denote the parallel (horizontal) composition of natural transformations: given $a:F \to F'$ and $b:G \to G'$, where all functors are endofunctor of $\A$, we have $a \ast b: FG \to F'G'$ given by $a \ast b = aG' \cdot Fb = F'b \cdot aG$. Recall also the interchange law:
			\begin{equation*}
			(c \ast d) \cdot (a \ast b) = (c \cdot a) \ast (d \cdot b).
			\end{equation*}
		\end{enumerate}
	\end{notation}

	\begin{proposition} \label{\labelauthor:pro:P-lim}
		The forgetful functor of Monad ($\A$) creates limits.
	\end{proposition}
	
	\subsection*{Remark}
		Recall that creation of limits means that for every diagram $D$ in Monad ($\A$) with a limit cone $p_d:T \to WDd$ of the underlying diagram in $[\A,\A]$ there exists a unique structure of a monad on $T$ for which each $p_d$ is a monad morphism. Moreover, the resulting cone is a limit in Monad ($\A$).
		\begin{proof}
			For the given diagram
			\begin{equation*}
				D: \D \to \textnormal{ Monad }(\A)
			\end{equation*}
			denote the objects by
			\begin{equation*}
				Dd = (T_d, \mu_d, \eta_d) \qquad (d \in \obj \D).
			\end{equation*}
			Given a limit cone $p_d: T \to T_d$, the unit of the monad on $T$ is, necessarily, the unique natural transformation
			\begin{equation*}
				\eta^T: \Id \to T \textrm{ with } p_d \cdot \eta^T = \eta_d \qquad(d \in \obj \D).
			\end{equation*}
			(Recall that $p_d$'s are required to preserve unit.) And the multiplication $\mu^T: T \cdot T \to T$ is, necessarily, the unique natural transformation for which the squares
			\begin{equation*}
				\xymatrix{
					T \cdot T \ar[r]^{\mu^T} \ar[d]_{p_d \ast p_d}
					&
					T \ar[d]^{p_d}
					\\
					T_d \cdot T_d \ar[r]_{\mu_d}
					&
					T_d
				}
			\end{equation*}
			commute for all $d \in \obj \D$. The verification of the monad axioms is easy. To verify that this is a limit cone, let $q_d:(S, \mu^S, \eta^S) \to (T_d, \mu_d, \eta_d)$ be a cone of $D$. There exists a unique natural transformation $g:S \to T$ with $q_d = p_d \cdot q(d \in \obj \D)$. It is a monad morphism. Indeed, the axiom $q \cdot \eta^S = q^T$ follows, since $(p_d)$ is a monocone, from
			\begin{equation*}
				p_d \cdot (q \cdot \eta^S) = q_d \cdot \eta^S = \eta_d = p_d \cdot \eta^T.
			\end{equation*}
			Analogously, the axiom $q \cdot \mu^S = \mu^T \cdot q \ast q$ follows from
			\begin{equation*}
				p_d \cdot (\mu^T \cdot (q \ast q)) = \mu_d \cdot (p_d \ast p_d) \cdot (q \ast q) = p_d \cdot \mu^T \cdot (q \ast q)
			\end{equation*}
		\end{proof}

	\begin{corollary}
		Limits of monads over a complete category $\A$ are computed object-wise (on the level of $\A$).
	\end{corollary}
	
	\begin{proposition} \label{\labelauthor:pro:Pro2.5}
		The forgetful functor of Monad ($\A$) creates absolute coequalizers.
	\end{proposition}
	
	\subsection*{Remark}
		Recall that this means that given a parallel pair of monad morphisms $p,q: \mathds{S \to T}$ whose coequalizers in $[\A,\A]$
		\begin{equation*}
			\xymatrix{
				S \ar@/^/[r]^{p} \ar@/_/[r]_{q}
				&
				T \ar[r]^c
				&
				C
			}
		\end{equation*}
		is absolute (that is, preserved by every functor with domain $[\A,\A]$), there exists a unique monad structure on $C$ making $c$ a monad morphism. Moreover, $c$ is a coequalizer of $p$ and $q$ in Monad $(\A)$.
		
		\begin{proof}
			The unit of $C$ is, necessarily,
			\begin{equation*}
				\eta^C = c \cdot \eta^T.
			\end{equation*}
			To define the multiplication $\mu^C: C \cdot C \to C$, use the endofunctor of $[\A,\A]$ defined by $ X \mapsto X \cdot X$ on objects and by $f \mapsto f \ast f$ on morphisms. Since $c \ast c$ is the coequalizer of $p \ast p$ and $q \ast q$, we have a unique $\mu^C$ for which $c$ preserves multiplication:
			\begin{equation*}
				\xymatrix{
					S \cdot S \ar@/^/[r]^{p \ast p} \ar@/_/[r]_{q \ast q} \ar[d]_{\mu^S}
					&
					T \cdot T \ar[r]^{c \ast c} \ar[d]_{\mu^T}
					&
					C \cdot C \ar[d]^{\mu^C}
					\\
					S \ar@/^/[r]^{p} \ar@/_/[r]_{q}
					&
					T \ar[r]_c
					&
					C
				}
			\end{equation*}
			The verification that $(C, \eta^C, \mu^C)$ is a monad and $c$ is a coequalizer in Monad ($\A$) is easy.
		\end{proof}
		
		\begin{definition} \label{\labelauthor:def:D-fix}
			An object $Z$ is a \emph{fixpoint} of an endofunctor $H$ if $HZ \simeq Z$, and it is a \emph{pre-fixpoint} of $H$ if $HZ$ is a subobject of $Z$.
			
			We say that $H$ has \emph{arbitrarily large pre-fixpoints} provided that for every object $X$ there exists a pre-fixpoint $Z$ of $H$ with $Z \simeq Z + X$.
		\end{definition}
		
		\begin{example}
			A monos-preserving endofunctor $H$ of the category $\Set^S$ of many-sorted sets has arbitrarily large pre-fixpoints iff for every cardinal $\alpha$ there exists a pre-fixpoint of $H$ all components of which have at least $\alpha$ elements.
		\end{example}		
		
		\begin{notation}
			\begin{enumerate}[(a)]
				\item For an endofunctor $H$ of $\A$ an \emph{algebra} is a pair $(A,a)$ consisting of an object $A$ and $a$ morphism $a:HA \to A$. Homomorphisms of algebras are defined by the usual commutative square. The resulting category is denoted by $\Alg H$.
				\item $\mu H$ denotes the initial algebra (if it exists). By Lambek's Lemma \cite{\labelauthor:bib:l} its algebra structure is invertible, thus, $\mu H$ is a fixpoint of $H$.
				\item If $H$ has \emph{free algebras}, i.\,e., the forgetful functor $\Alg H \to \A$ has a left adjoint, then $\bkF_H$ denotes the corresponding monad over $\A$. And $\hat{\eta}: \Id \to \bkF_H$ denotes its unit, whose components are the universal arrows of the free algebras.
			\end{enumerate}
		\end{notation}
		
		\begin{lemma} \label{\labelauthor:lem:L-27}
			Every accessible endofunctor of a cocomplete category with monic coproduct injections has arbitrarily large pre-fixpoints.
		\end{lemma}
		
		\begin{proof}
			If $H$ is accessible, then every object $B$ generates a free $H$-algebra $\bar{B}$ and $\bar{B} = B + H\bar{B}$, see \cite{\labelauthor:bib:a}. Given an object $A$ let $B$ be an infinite copower of $A$. Then the equality $A+B \simeq B$ implies $A + \bar{B} \simeq \bar{B}$, and $\bar{B}$ is a pre-fixpoint.
		\end{proof}		
		
		\begin{theorem}[Barr \cite{\labelauthor:bib:b}] \label{\labelauthor:thm:T-B}
			If an endofunctor $H$ has free algebras, then $\bkF_H$ is a free monad on $H$. The converse holds whenever the base category is complete.
		\end{theorem}
		
		\begin{example} \label{\labelauthor:ex:E-P}
			The power-set functor $\P$ has no fixpoint, hence, it does not generate a free monad.
		\end{example}
		
		\begin{construction}[see \cite{\labelauthor:bib:a}] \label{\labelauthor:con:C-W}
		For every object $X$ of $\A$ define the \emph{free-algebra chain} $W: \Ord \to \A$ (with objects $W_i$ and morphisms $w_{i,j}: W_i \to W_j$ for all ordinals $i \leq j$) uniquely up to natural isomorphism by the following transfinite induction:
		
		The objects are given by
		\begin{align*}
			&W_o = X\\
			&W_{i+1} = X + HW_i,
		\end{align*}
		and
		\begin{equation*}
			W_j = \colim_{i < j} W_i \textrm{ for limit ordinals }j.
		\end{equation*}
		The morphisms are as follows:
		\begin{align*}
			&w_{0,1}: X \to X + HX, \textrm{ coproduct injection}\\
			&w_{i+1,j+1} = \id_X + Hw_{i,j}
		\end{align*}
		and
		\begin{equation*}
			(w_{i,j})_{i<j} \textrm{ is a colimit cocone (for limit ordinals $j$).}
		\end{equation*}
		Whenever this chain \emph{converges} after $i$ steps, i.\,e., all connecting maps $w_{i,j}$ are isomorphisms, then as proved in \cite{\labelauthor:bib:a},
		\begin{equation*}
			W_i = F_HX
		\end{equation*}
		is the free algebra on $X$. More detailed, the two components of
		\begin{equation*}
			(w_{i,i+1})^{-1}: X + HW_i \to W_i
		\end{equation*}
		are the universal arrow and the algebra structure of $W_i$, respectively.
	\end{construction}
	
	\begin{definition} \label{\labelauthor:def:D-stable}
		A cocomplete category is said to have \emph{stable monomorphisms} if
		\begin{enumerate}[(a)]
			\item coproducts of parallel collections of monomorphisms are monic
		\end{enumerate}
		and
		\begin{enumerate}[(b)]
			\item colimits of chains of monomorphisms consist of monics, and the factorizing map of every cocone of monics is monic.
		\end{enumerate}
	\end{definition}
	
	\begin{example}
		Sets, graphs, posets, many-sorted sets and almost all "usual" varieties of algebras have stable monomorphisms. All presheaf categories have stable monomorphisms.
		
		Condition (b) implies that the unique morphism from $0$ to any given object is monic (since $0$ is the colimit of the empty chain). Thus rings are an example of a variety not having stable monomorphisms. Indeed, the initial ring is the ring $\bkZ$ of integers, and not all ring homomorphisms with this domain are monic.
	\end{example}
	
	\begin{theorem}[See \cite{\labelauthor:bib:takr}] \label{\labelauthor:teo:t-takr}
		Let $H$ be an endofunctor of a cocomplete category with stable monomorphisms. If $H$ preserves monomorphisms, the following conditions are equivalent:
		\begin{enumerate}[(1)]
			\item $H$ has free algebras
			\item for every object $X$ the free-algebra chain converges
		\end{enumerate}
		and
		\begin{enumerate}[(3)]
			\item for every object $X$ there exists an object $Z$ with
			\begin{equation*}
				HZ+X \textnormal{ a subobject of } Z.
			\end{equation*}
		\end{enumerate}		
	\end{theorem}
	
%	\begin{definition}
%		A pre-fixpoint of an endofunctor $H$ is an object $Z$ such that $HZ$ is a subobject of $Z$. We say that $H$ \emph{has arbitrarily large pre-fixpoints} if for every object $X$ there exist a pre-fixpoint $Z$ with $Z+X \simeq Z$.
%	\end{definition}
	
	\begin{corollary} \label{\labelauthor:cor:C-free}
		Let $\A$ be a cocomplete category with stable monomorphisms. Every monos-preserving endofunctor with arbitrarily large pre-fixpoints generates a free monad.
	\end{corollary}	
	Indeed, we verify Condition (3) above: choose a pre-fixpoint $Z$ with $Z \simeq Z+X$ to get $HZ + X$ as a subobject of $Z+X \simeq Z$.

	\begin{remark} \label{\labelauthor:rem:r-mono}
		Under the assumptions of the above theorem the free monad $\bkF_H$ preserves monomorphisms. Indeed, let $m: X \to X'$ be a monomorphism. Denote by $W_i'$ the free-algebra chain above for $X'$. It is easy to see that we get a natural transformation
		\begin{equation*}
			m_i : W_i \to W_i' \qquad (i \in \Ord)
		\end{equation*}
		by
		\begin{align*}
			m_0 &= m:X \to X'\\
			m_{i+1} &= m + Hm_i : X + HW_i \to X' + HW_i'
		\end{align*}
		and
		\begin{equation*}
			m_j = \colim_{i<j} m_i \textnormal{ for limit ordinals }j.
		\end{equation*}
		An easy transfinite induction shows that $m_i$ is monic for every $i$: in the isolated step use the preservation of monics by $H$.
		
		We know from the above theorem that for some ordinal $i$ we have
		\begin{equation*}
			F_HX = W_i \textnormal{ and } F_HX' = W_i'.
		\end{equation*}
		For this ordinal we then also have
		\begin{equation*}
			F_H m = m_i
		\end{equation*}
		(which follows by an easy inspection of the proof of the above theorem). Thus, $F_Hm$ is monic.
	\end{remark}
	
	\begin{remark} \label{\labelauthor:rem:r-plus}
		If free $H$-algebras exist, the free monad $\bkF_H$ fulfils
		\begin{equation*}
			F_H = H \cdot F_H + \Id, \textnormal{ with } \hat{\eta} \textnormal{ as the right-hand injection.}
		\end{equation*}
		Indeed, for every object $X$ let $\bar{X}$ the free algebra on $\hat{\eta}_X: X \to \bar{X}$ with the algebra structure $\varphi_X: H \bar{X} \to \bar{X}$. Then $\bar{X} = H \bar{X} + X$ since $[\varphi_X, \hat{\eta}_X] : H \bar{X} + X \to \bar{X}$ is an isomorphism. (This is Lambek's Lemma applied to $H(-) + X$.) Since $F_H X = \bar{X}$, we see that the natural transformations $\varphi: HF_H \to F_H$ and $\hat{\eta}: \Id \to F_H$ form coproduct injections of $F_H = H \cdot F_H + \Id$.
	\end{remark}
	
\section{Set-Like Base Categories}
	For the base categories $\A$ such as
	\begin{enumerate}
		\item[] $\Set$ or $\Set^S$ (many-sorted sets)
		\item[] K-Vec (vector spaces)
		\item[] $\Set_\ast$ (sets and partial functions)
	\end{enumerate}
	we prove that the category of monads has coequalizers and cointersections. And it has colimits of every diagram of monos-preserving monads that posses arbitrarily large joint pre-fixpoints. In case of coproducts over $\A = \Set$ that last condition was proved to be "almost" necessary in \cite{\labelauthor:bib:ablm}: a collection of nontrivial monads over $\Set$ has a coproduct iff they posses arbitrarily large joint fixpoints. We explain this in more detail in the next section devoted to coproducts of separated monads.
	
	\begin{assumptions}
		Thoughout this section $\A$ denotes a category which has
		\begin{enumerate}[(a)]
			\item limits and colimits
			\item stable monomorphisms (see Definition \ref{\labelauthor:def:D-stable})
		\end{enumerate}
		and
		\begin{enumerate}[(c)]
			\item split epimorphisms.
		\end{enumerate}
	\end{assumptions}
	
	\begin{remark} \label{\labelauthor:rem:R-cow}
		\begin{enumerate}[(a)]
			\item $\A$ is cowellpowered: every object $X$ has only a set of quotients because $X$ has only a set of idempotent endomorphisms. Indeed, for every quotient $e:X \to Y$ choose a splitting $i: Y \to X$ and get an idempotent i.\,e, then two epimorphisms with the same idempotent yield the same quotient.
			\item $\A$ has (strong epi, mono)-factorizations of morphisms since every co\-well\-po\-wered, cocomplete category does, see \cite{\labelauthor:bib:ahs}, 15.17.
		\end{enumerate}
	\end{remark}
	
	\begin{lemma} \label{\labelauthor:lem:L-fact}
		Monad ($\A$) has (strong epi, mono)-factorization of morphisms, and every strong epimorphism has all components epic.
	\end{lemma}
	
	\begin{proof}
		We prove that every monad morphism $f: \bkS \to \bkR$ has a factorization $f = m \cdot e$ in Monad ($\A$) where $m$ has monic components and $e$ has (split) epic ones. It follows easily from Proposition \ref{\labelauthor:pro:P-lim} that $m$ is a monomorphism in Monad ($\A$) and $e$ is a strong epimorphism.
		
		Indeed, start with a factorization of every $f_A$ in $\A$ as $SA \xrightarrow{e_A} RA \xrightarrow{m_A} TA$ with $e_A$ split epic and $m_A$ monic in $\A$. Then the diagonal fill-in makes $R$ an endofunctor with natural transformations $e:S \to R$ and $m:R \to T$. The monad unit of $R$ is $\mu^R = e\cdot\eta^\bkS: \Id \to R$. And the monad multiplication is given by the following diagonal fill-in:
		\begin{equation*}
			\xymatrix@+25pt{
				SSA \ar@{->>}[r]^{e_{A^\ast}e_A} \ar[d]_{\mu^S_A}
				&
				RRA \ar[r]^{m_{A^\ast}m_A} \ar@{-->}[d]_{\mu^R_A}
				&
				TTA \ar[d]^{\mu^T_A}
				\\
				SA \ar[r]_{e_A}
				&
				RA \ar[r]_{m_A}
				&
				TA
			}
		\end{equation*}
		This is well-defined because $e_A \ast e_A = e_{RA} \cdot R e_A$ is a epimorphism. To verify the unit axioms $\mu^R \cdot \eta R = \id$, consider the following diagram:
		\begin{equation*}
			\xymatrix@+15pt{
				SA \ar[r]^{e_A} \ar[d]_{\eta^S_{SA}}
				&
				RA \ar[r]^{m_A} \ar[d]_{\eta^R_{RA}}
				&
				TA \ar[d]^{\eta^T_{TA}}
				\\
				SSA \ar[r]^{e_A \ast e_A} \ar[d]_{\mu^S_A}
				&
				RRA \ar[r]^{m_A \ast m_A} \ar[d]_{\mu^R_A}
				&
				TTA \ar[d]^{\mu^T_A}
				\\
				SA \ar[r]_{e_A}
				&
				RA \ar[r]_{m_A}
				&
				TA
			}
		\end{equation*}
		Its outward square commutes since $\bkS$ and $\bkT$ both satisfy the corresponding axiom. Naturality of $\eta^S$ implies that the upper left-hand square commutes:
		\begin{equation*}
			e_{RA} \cdot Se_A \cdot \eta^S_{SA} = e_{RA} \cdot \eta^S_{RA} \cdot e_A = \eta^R_A \cdot e_A.
		\end{equation*}
		Analogously for the upper right-hand square. Consequently, the diagonal passage from $SA$ to $TA$ in the above diagram satisfies (due to $\mu^T_A \cdot \eta_{TA} = \id)$ the equality
		\begin{equation*}
			m_A \cdot (\mu^R_A \cdot \eta^R_{RA}) \cdot e_A = m_A \cdot e_A.
		\end{equation*}
		Since $n_A$ is strongly monic and $e_A$ epic, this implies $\mu^R_A \cdot \eta^R_{RA} = \id$. \smallskip
		
		The verification of the other unit axiom $\mu^R \cdot R \eta^R = \id$ is analogous. \smallskip
		
		The proof of the associativity
		\begin{equation*}
			\mu^R \cdot R \mu^R = \mu^R \cdot \mu^R R
		\end{equation*}
		follows from the following diagram:
		\begin{equation*}
			\xymatrix@+25pt{
				SSSA \ar[r]^{e_A\ast e_A\ast e_A} \ar@<-2pt>[d]_{\mu^S_{SA}} \ar@<+2pt>[d]^{\mu^S_A}
				&
				RRRA \ar[r]^{m_A \ast m_A \ast m_A} \ar@<-2pt>[d]_{\mu^R_{RA}} \ar@<+2pt>[d]^{R \mu^R_A}
				&
				TTTA \ar@<-2pt>[d]_{\mu^T_{TA}} \ar@<+2pt>[d]^{T \mu^T_A}
				\\
				SSA \ar[r]^{e_A \ast e_A} \ar[d]_{\mu^S_A}
				&
				RRA \ar[r]^{m_A \ast m_A} \ar[d]_{\mu^R_A}
				&
				TTA \ar[d]^{\mu^T_A}
				\\
				SA \ar[r]_{e_A}
				&
				RA \ar[r]_{m_A}
				&
				TA
			}
		\end{equation*}
		We only need to check that the epimorphism $e_A \ast e_A \ast e_A$ merges the above parallel pair. Since $m_A$ is a monomorphism and the outward square of the above diagram is the following commutive square
		\begin{equation*}
			\xymatrix@+15pt{
				SSSA \ar[r]^{f_A \ast f_A \ast f_A} \ar[d]_{\mu^S_A \cdot S \mu^S_A}
				&
				TTTA \ar[d]^{\mu^T_A \cdot T \mu^T_A}
				\\
				SA \ar[r]_{f_A}
				&
				TA
			}
		\end{equation*}
		the accociativity of $\mu^S$ and $\mu^T$ clearly implies that of $\mu^R$.
	\end{proof}
	Recall from Definition \ref{\labelauthor:def:D-fix} the concept of arbitrarily large pre-fixpoints of an endofunctor. Here is a "collective" version:
	
	\begin{definition} \label{\labelauthor:def:D-34}
		A collection $F_i\,(i \in I)$ of endofunctors is said to have \emph{arbitrarily large joint pre-fixpoints} if for every object $A$ and every cardinal $\alpha > 0$ there exists a joint pre-fixpoint $X$ such that $X+A \simeq X \simeq \coprod_\alpha X$.
	\end{definition}
	
	\begin{example}
		For categories $\Set$ and K-Vec or $\Set_\ast$ this means that for every cardinal $\alpha$ there exists a joint pre-fixpoint of cardinality at least $\alpha$. (In K-Vec use the fact that for infinite cardinals $\alpha \geq \card K$ dimension $\alpha$ is equivalent to cardinality $\alpha$.)
		
		For many-sorted sets, $\Set^S$, this means that for every cardinal $\alpha$ there exists a joint pre-fixpoint whose components have cardinalities at least $\alpha$. 
	\end{example}
	
%	In the following proposition we assume that the base category in cocomplete and
%	\begin{enumerate}[(a)]
%		\item coproduct injections are monic
%	\end{enumerate}
%	and
%	\begin{enumerate}[(b)]
%		\item for some cardinal $\lambda_0$ colimits of $\lambda$-chains, where $\lambda \geq \lambda_0$ is a regular cardinal, preserve monomorphisms.
%	\end{enumerate}
%	Thus for example every locally presentable category with monic coproduct injections satisfies this, see \cite{\labelauthor:bib:AR}, 1.62.
	
	\begin{proposition}
		Every collection of accessible endofunctors has arbitrarily large joint pre-fixpoints.
	\end{proposition}
	
	\begin{proof}
		If $H_r (r \in R)$ are accessible endofunctors, then so is $H = \coprod_{r \in R} H_r$. And every pre-fixpoint of $H$ is a joint pre-fixpoint of all $H_r$. Thus our task is for a given object $A$ and an infinite cardinal $\alpha$, to find a pre-fixpoint $X$ of $H$ with $X \simeq A + X \simeq \alpha \cdot X$. 
%		Let $\lambda \geq \lambda_0$ be a regular cardinal such that $H$ preserves $\lambda$-directed colimits. A coproduct of $\alpha$ copies of an object $A$ is denoted by $\alpha \bullet A$. We define a smooth $\lambda$-chain $X_i ( i \leq \lambda)$, i.\,e., one which on limit ordinals is given a colimits, by the following transfinite induction:
%		$$
%			W_0 = \alpha \bullet A \textnormal{ and } W_{i+1} = \alpha \bullet A + \alpha \bullet HW_i
%		$$
%		The connecting morphisms are defined by letting
%		$$
%			w_{01}: \alpha \bullet A \to \alpha \bullet A + H (\alpha \bullet A)
%		$$
%		be the left-hand coproduct injection, and given $w_{ji}: W_i \to W_j$ defining
%		$$
%			w_{j+1, i+1} = \id_{\alpha \bullet A} + H w_{j,i}.
%		$$
%		It is obvious that the colimit
%		$$
%			X = \colim_{i < \lambda} W_i
%		$$
%		fulfils $A + X \simeq X$ and $\alpha \bullet X \simeq X$ (due to canonical isomorphisms $A + X_i \simeq X_i \simeq \alpha \bullet X_i$).
%		
%		To prove that $X$ is a pre-fixpoint of $H$, choose an arbitrarily element of $\alpha$ and denote by
%		$$
%			m_i: HW_i \to W_{i+1} = \alpha \bullet A + \alpha \bullet HW_i
%		$$
%		the corresponding coproduct injection. This is clearly a natural transformation from $HW_i$ to $W_{i+1}~(i < \lambda)$ consisting of monomorphisms. Due to $\lambda\geq \lambda_0$ we get a monomorphism
%		$$
%			\colim_{l < \lambda} m_i : HX =\colim_{i < \lambda} HX_i \to X
%		$$
		The copower $\alpha \bullet H$ of $\alpha$ copies of $H$ is accessible, thus, it has a free algebra on $B = \alpha \bullet A$. As in the proof of Lemma \ref{\labelauthor:lem:L-27} this free algebra $\bar{B}$ fulfils
		$$
			\bar{B} = \alpha \bullet A + \alpha \bullet H \bar{B} = \alpha \bullet (A + H \bar{B})
		$$
		Obviously, $H \bar{B}$ is a subobject of $\alpha \cdot H \bar{B}$, hence, a pre-fixpoint of $H$. And $\bar{B} \simeq A + \bar{B} \simeq \alpha \bullet \bar{B}$.
	\end{proof}
	
	\begin{theorem} \label{\labelauthor:thm:T-33}
		Every small collection of monos-preserving monads with arbitrarily large joint pre-fixpoints has a coproduct in Monad ($\A$).
	\end{theorem}
	
	\begin{proof}
		Let $\bkS_i = (S_i, \mu_i, \eta_i), i \in I$, be such a collection. Then the endofunctor $S = \coprod_{i \in I} S_i$ preserves monomorphisms. And it has arbitrarily large pre-fixpoints: given an object $A$ find $X$ with $S_i X \rightarrowtail X$ for all $i \in I$ and $X \simeq X+A \simeq \coprod_I X$ to get
		\begin{equation*}
			SX = \coprod_{i \in I} S_i X \rightarrowtail \coprod_I X \simeq X.
		\end{equation*}
		By Corollary \ref{\labelauthor:cor:C-free} the functor $S= \coprod_{i \in I} S_i$ generates a free monad $\bkF_S$ with the universal arrow $\hat{\eta}: S \to FS$; the coproduct injections are denoted by
		\begin{equation*}
			v_i: S_i \to S \qquad (i \in I)
		\end{equation*}
		 The forgetful functor Monad $(\A) \to [\A,\A]$ creates limits, see Proposition \ref{\labelauthor:pro:P-lim}, and we conclude that for the slice category $\bkF_S /$ Monad ($\A$) the the corresponding forgetful functor
		\begin{equation*}
			U : \bkF_S / \textnormal{ Monad }(\A) \to F_S / [A,A]
		\end{equation*}
		also creates limits. Now consider an arbitrarily cocone $f = (f_i)$ consisting of monad morphisms $f_i: \bkS_i \to \bkT_f \, (i \in I)$. The functor $[f_i]: S \to T_f$ generates uniquely a monad morphism $\bar{f}: \bkF_S \to \bkT_f$ with $\bar{f} \cdot \hat{\eta} = [f_i]$ that we factorize as in Lemma \ref{\labelauthor:lem:L-fact}
		\begin{equation*}
			\xymatrix{
				S \ar[rrd]^{[f_i]} \ar[d]_{\eta^S}
				\\
				\bkF_S \ar[rr]_{\bar{f}} \ar@{->>}[rd]_{e_f}
				&&
				\bkT_f
				&
				\\
				&
				\bkR_f \ar@{>->}[ru]_{m_f}
			}
		\end{equation*}
		We get a (possibly large) collection of objects $(e_f, R_f)$ of the slice category $\bkF_S /$ Monad ($\A$). This collection has a product in $F_S / [\A,\A]$. Indeed, recall from Remark \ref{\labelauthor:rem:R-cow} that $\A$ is cowellpowered, and for every object $A$ form the meet of $(e_f)_A: F_SA \to R_fA$ ranging through all cocones $f$. Let $e_A: F_SA \to RA$ be meet, thus for every cocone $f$ we have a morphism
		\begin{equation*}
			q^A_f: RA \to R_fA \textnormal{ with } (e_f)_A = q^A_f \cdot e_A.
		\end{equation*}
		The resulting functor $R$ and natural transformations $q_f: R \to R_f$ form a product of all $e_f$ in $F_S / [\A,\A]$. Consequently, there exists a product $(e, \bkR)$ of the objects $(e_f, \bkR_f)$ in $\bkF_S /$ Monad ($\A$) as $f$ ranges through all cocones: see Proposition \ref{\labelauthor:pro:P-lim}. For the projections $q_f: \bkR \to \bkR_f$ define
		\begin{equation*}
			p_f = m_f \cdot q_f : \bkR \to \bkT_f.
		\end{equation*}
		Then $\bar{f} = m_f \cdot e_f = m_f \cdot q_f \cdot e = p_f \cdot e$ implies
		\begin{equation*}
			f_i = f \cdot v_i = p_f \cdot e \cdot \hat{\eta} \cdot v_i:
		\end{equation*}
		\begin{equation} \label{\labelauthor:eqn:eqn31}
			\vcenter{
			\xymatrix{
				S_i \ar[d]_{v_i} \ar[rddd]^{f_i}
				\\
				S \ar[d]_{\hat{\eta}}
				\\
				\bkS \ar[rd]_{e_f} \ar[d]_{e}
				\\
				\bkR \ar[r]_{p_f} \ar@{<-} `l[u] `[uuu]^{u_i} [uuu]
				&
				\bkT_f 
			}
			}
		\end{equation}
		We claim that $\bkR$ is the coproduct of $\bkS_i \, (i \in I)$ in Monad ($\A$) with respect to
		\begin{equation*}
			u_i = e \cdot \hat{\eta} \cdot m_i : \bkS_i \to \bkR \qquad (i \in I).
		\end{equation*}
		\begin{enumerate}[(a)]
			\item Each $u_i$ is a monad morphism. This follows from the fact that $(p_f)$ is a collectively monic cone in $[\A, \A]$ and each $f_i$ is a monad morphism. Indeed, the condition $u_i \cdot \eta_i = \eta^R$ follows from
			\[
			\begin{array}{rll}
				p_f \cdot (u_i \cdot \eta_i) &= f_i \cdot \eta_i &\textnormal{ see (\ref{\labelauthor:eqn:eqn31})}\\
				&= \eta^{\bkR_f} & f_i \textnormal{ a monad morphism}\\
				&= p_f \cdot \eta^\bkR & p_f \textnormal{ a monad morphism.}
			\end{array}
			\]
		\end{enumerate}
		The verification of the condition 
		\begin{equation*}
			\mu_i \cdot u_i = \mu^R \cdot u_i \ast u_i = \mu^R \cdot Ru_i \cdot u_iS_i
		\end{equation*}
		follows from the following diagram
		\begin{equation*}
			\xymatrix{
				S_iS_i \ar[rrrr]^{\mu_i} \ar[rd]^{u_iS_i} \ar[dd]_{f_iS_i}
				&&&&
				S_i \ar[ldddd]^{f_i} \ar[ldd]_{u_i}
				\\
				&
				RS_i \ar[ld]^{p_fS_i} \ar[rd]^{Ru_i}
				\\
				T_fS_i \ar[rrd]^{T_fu_i} \ar[rrdd]_{T_ff_i}
				&&
				RR \ar[r]^{\mu^R} \ar[d]^{p_fR}
				&
				R \ar[dd]_{p_f}
				\\
				&&
				T_fR \ar[d]^{T_fp_f}
				\\
				&&
				T_fT_f \ar[r]_{\mu^{T_f}}
				&
				T_f
			}
		\end{equation*}
		All the inner parts but the upper one (to be proved commutative) commute: recall $f_i = p_f \cdot u_i$, use the fact that $p_f$ is a monad morphism for the lower square, and use the naturality of $p_f$ for $p_fR \cdot Ru_i = T_f u_i \cdot p_f S_i$. Since $f_i$ is a monad morphism, the outward square also commutes. This, together with the collective monicity of all $p_f$'s, proves that the upper square commutes.
		\item For every cocone $f= (f_i)_{i \in I}$ the monad morphism $p_f$ is the desired factorization: $f_i = p_f \cdot u_i$, see (\ref{\labelauthor:eqn:eqn31}). This is unique since whenever $r: \bkR \to \bkT_f$ is a monad morphism with $f_i = r \cdot u_i$ for all $i$, then $r \cdot e \cdot \hat{\eta} = f = p_f \cdot e \cdot \hat{\eta}$ which implies $r \cdot e = p_f \cdot e$ by the universal property of $\hat{\eta}$; hence $r = p_f$ since $e$ is epic.
	\end{proof}
	
	\begin{remark} \label{\labelauthor:rem:R-Kelly}
		\begin{enumerate}[(a)]
			\item Kelly described colimits of monads, see \cite{\labelauthor:bib:Ke}, Section 27 as follows:
		
		Let $D$ be a diagram in Monad $(\A)$ with objects $\mathds{T}_i = (T_i, \mu_i, \eta_i)$ for $i \in I$.
		
		Form the category $\C_D$ of all pairs $(A, (a_i)_{i \in I})$ where $A$ is an object of $\A$ and $a_i: T_iA \to A$ is an Eilenberg-Moore algebra for $\mathds{T}_i~(i \in I)$ such that for every connecting morphism $f: i \to j$ of the indexing category the triangle
		\begin{equation} \label{\labelauthor:eqn:eqn32}
			\vcenter{
			\xymatrix{
				T_iA \ar[r]^{a_i} \ar[rd]_{(Df)_A}
				&
				A
				\\
				&
				T_jA \ar[u]_{a_j}
			}}
		\end{equation}
		commutes. The morphisms of $\C_D$ are the morphisms of $\A$ which are algebra homomorphisms for every $\mathds{T}_i$. We have the obvious forgetful functor
		\begin{equation*}
			U_D: \C_D \to \A.
		\end{equation*}
		Kelly proved that if $U_D$ has a left adjoint, then the corresponding monad on $\A$ is a colimit of $D$ in Monad $(\A)$. The converse also holds if $\A$ is a complete category.	
		\end{enumerate}
	\end{remark}
	
	\begin{theorem} \label{\labelauthor:thm:thm3.9}
		Every diagram with a weakly terminal object has a colimit in Monad ($\A$). In particular, Monad ($\A$) has coequalizers.
	\end{theorem}
	
	\begin{proof}
		Let $D: \D \to$ Monad $(\A)$ be a diagram with objects $\mathds{T}_i = (T_i, \mu_i, \eta_i)$ for $i \in I$, and let $\mathds{T}_j$ be weakly terminal, i.e., for every $i \in I$ there exists a connecting morphism $f: \mathds{T}_i \to \mathds{T}_j$ in $D$. 
		\begin{enumerate}[(a)]
			\item Form the full subcategory $\C$ of $\A^{\mathds{T}_j}$ of all algebras $a: T_j A \to A$ for $\bkT_j$ such that for every pair $f,g:\mathds{T}_i \to \mathds{T}_j $ of connecting morphisms of $D~(i \in I)$ we have
			\begin{equation} \label{\labelauthor:eqn:Equal}
				a \cdot f_A = a \cdot g_A
			\end{equation}
			This category is closed in $\A^{\mathds{T}_j}$ under products, which easily follows from the forgetful functor $U^{\mathds{T}_j}$ creating limits. It is also closed under subalgebras.
			More precisely, let $m: (A,a) \to (B,b)$ be a homomorphism in $\A^{\mathds{T}_j}$ with $m$ monic in $\A$. If $(B, b)$ lies in $\C$, then so does $(A, a)$:
			\begin{equation*}
				\xymatrix@+30pt{
					T_iA \ar@/^/[r]^{f_A} \ar@/_/[r]_{g_A} \ar[d]_{T_im}
					&
					T_jA \ar[r]^a \ar[d]^{T_jm}
					&
					A \ar[d]^m
					\\
					T_iB \ar@/^/[r]^{f_B} \ar@/_/[r]_{g_B}
					&
					T_jB \ar[r]_b
					&
					B
				}
			\end{equation*}
%			Indeed, $b \cdot f_B = b \cdot g_B$ implies that $m$ merges the upper parallel pair. Thus we get $a \cdot f_A = a \cdot g_A$. 		
%			\item The category $\C$ is reflective in $\A^{\bkT_j}$. Indeed, since $\C$ is closed under products and strong subobjects, i.\,e., subalgebra, it is sufficient to verify that $\A^{\bkT_j}$ is complete, wellpowered and cowellpowered: see \cite{\labelauthor:bib:ahs}, 16.9. The first two

			Since the forgetful functor $U^{\mathds{T}_j}$ creates limits, the category $\A^{\mathds{T}_j}$ is complete and wellpowered. Let us prove that it is also cowellpowered. Given a factorization of a homomorphism $h:(A,a) \to (B,b)$ in $A^{\mathds{T}_j}$ as a strong epimorphism $e:C \to B$ followed by a monomorphism $m: C \to B$ in $\A$, the diagonal fill-in makes $e$ and $m$ homomorphisms:
		\begin{equation*}
			\xymatrix{
				T_jA \ar[r]^a \ar@{->>}[d]_{T_je}
				&
				A \ar@{->>}[d]^e
				\\
				T_jC \ar@{-->}[r]^{c} \ar[d]_{T_jm}
				&
				C \ar@{>->}[d]^m
				\\
				T_jB \ar[r]_b
				&
				B
			}
		\end{equation*}	
		Thus, if $h$ is a strong epimorphism in $A^{\mathds{T}_j}$ then $m$ is an isomorphism (recall that $U^{\mathds{T}_j}$ creates limits, thus, reflects isomorphisms), consequently, $h$ is an epimorphism in $\A$. Since $\A$ is cowellpowered (see Remark \ref{\labelauthor:rem:R-cow}) we conclude that $A^{\mathds{T}_j}$ is cowellpowered. 
		\item Every full subcategory $of \A^{\bkT_j}$ closed under products and subobjects is reflective, see \cite{\labelauthor:bib:ahs}, 16.9. Thus, the obvious forgetful functor $U: \C \to \A$ has a left adjoint.
		
		The theorem now follows from Remark \ref{\labelauthor:rem:R-Kelly} and the fact that there exists an isomorphism $E$ of categories such that the triangle
		\begin{equation*}
			\xymatrix{
				\C_D \ar[rr]^E \ar[rd]_{U_D}
				&&
				\C \ar[ld]^{U}
				\\
				&
				\A		
			}
		\end{equation*}
		commutes. Indeed, $E$ is the "projection to $j$"
		\begin{equation*}
			E(A, (d_i)_{i \in I}) = (A, d_j).
		\end{equation*}
		From the triangles (\ref{\labelauthor:eqn:eqn32}) we deduce that $(A, d_j)$ satisfies (\ref{\labelauthor:eqn:Equal}). Thus, $E$ is a well-defined, faithful functor. It is surjective on objects: for every algebra $(A, a)$ in $\C$ define, given $i \in I$,
		\begin{equation*}
			a_i = a \cdot f_A:T_iA \to A \text{ for any connecting morphism }f:\mathds{T}_i \to \mathds{T}_j.
		\end{equation*}
		Then $a_i$ is well-defined due to (\ref{\labelauthor:eqn:Equal}) and, since $f$ is a monad morphism, $(A, a_i)$ is an Eilenberg-Moore algebra for $\mathds{T}_i$.		
		Finally, to prove that $E$ is an isomorphism, we verify that it is full. Let
		\begin{equation*}
			k:(A,a) \to (B,b)
		\end{equation*}
		be a homomorphism in $\C$. Then we need to prove that for every $i \in I$ this is a homomorphism from $(A, a_i)$ to $(B, b_i)$, where again $b_i = b \cdot f_B$. Use the following diagram
		\begin{equation*}
			\xymatrix@+30pt{
				T_iA \ar[r]^{f_A} \ar[d]_{T_i k}
				&
				T_jA \ar[r]^{a} \ar[d]_{T_jk}
				&
				A \ar[d]^k \ar@{<-} `u[l] `[ll]_{a_i} [ll]
				\\
				T_iB \ar[r]^{f_B}
				&
				T_jB \ar[r]^b
				&
				B \ar@{<-} `d[l] `[ll]^{b_i} [ll]
			}
		\end{equation*}
		\end{enumerate}
	\end{proof}

	\begin{corollary}
		Every diagram of monos-preserving monads with arbitrarily large joint pre-fixpoints has a colimit in Monad ($\A$).
	\end{corollary}

	Indeed, apply the usual construction of colimits as coequalizers of a parallel pair between coproducts; see \cite{\labelauthor:bib:ML}. Given a diagram $\D$ in Monad ($\A$) with monos-preserving objects $\bkS_i = (S_i, \mu_i, \eta_i)$ for $i \in I$ having arbitrarily large joint pre-fixpoints, then also every collection of monads indexed by $I \times J$, where $J$ is an arbitrarily set and $\bkS_i = \bkS_{(i,j)}$ for all $(i,j) \in I \times J$, has arbitrarily large joint pre-fixpoint. (Indeed, for every object $A$ and every cardinal $\alpha$ put $\alpha' = \alpha + \card J$. By applying Definition \ref{\labelauthor:def:D-34} to $A$ and $\alpha'$ for the former collection indexed by $I$, we get the required condition for the new collection.) Thus, those two coproducts needed to construct $\colim \D$ as a coequalizer in Monad ($\A$) exist.

	\begin{remark} \label{\labelauthor:def:rmk-311}
		Monad ($\A$) also has cointersections. That is, wide pushouts of strong epimorphisms $e_i: \bkT \to \bkS_i \, (i \in I)$. The proof is analogous to that of Theorem \ref{\labelauthor:thm:thm3.9}. Let $\C$ be the full subcategory of $\A^\bkT$ on all algebras $a:TA \to A$ for which $a$ factorized though each $(e_i)_A: TA \to S_i A$ factorizes though $a$. This subcategory is easily seen to be closed under products and subalgebras. And it is isomorphic to the category $\C_\D$ of Remark \ref{\labelauthor:rem:R-Kelly}. (Here we use the fact established in Lemma \ref{\labelauthor:lem:L-fact} that strong epimorphisms in Monad ($\A$) have epic components.) Thus, the cointersection of $c_i$ exists in Monad ($\A$).
	\end{remark}
	
	\begin{example}
		For the base category of graphs
		\begin{equation*}
			\Gra = \Set^{\rightrightarrows}
		\end{equation*}
		we present a parallel pair of monad morphisms having no coequalizer in Monad $(\Gra)$.
		
		For every graph $X=(V,E,s,t)$ with source and target maps $s,t: E \to V$ we denote by $X_e$ the set of all loops, i.\,e., the equalizer of $s$ and $t$. We construct two endofunctors $H,K: \Gra \to \Gra$ and two natural transformations $\sigma, \tau$: $H \to K$ such that for the coequalizer
		\begin{equation*}
			\xymatrix@1{			
			 	H \ar@/^/[r]^\sigma \ar@/_/[r]_\tau
			 	&
			 	K \ar[r]^{\rho}
			 	&
			 	L 
			 }
			 \quad \textrm{ in } [\Gra,\Gra]
		\end{equation*}
		$L$ does not generate a free monad, but $H$ and $K$ do. It follows immedeately that the monad morphisms
		\begin{equation*}
			\bar{\sigma}, \bar{\tau}: \mathds{F}_H \to \mathds{F}_K
		\end{equation*}
		corresponding to $\sigma$ and $\tau$ do not have a coequalizer in Monad $(\Gra)$: if $\mathds{S}$ were the codomain of such a coequalizer, then since $\bkF_{(-)}$ is a left adjoint, $\mathds{S}$ would clearly be a free monad on $L$.
		
		Let $\P$ denote the power-set functor. The endofunctor $H$ is defined on objects $X$ as follows:
		\begin{equation*}
			H(X) \textrm{ has vertices } \P(X_e) \textrm{ and no edges.}
		\end{equation*}
		The definition of $H$ on morphisms $g:X \to X'$ is as expected: $H(g)$ is the domain-codomain restriction of the edge function of $g$ to all loops. Analogously define $K$:
		$$
			K(X) \textnormal{ has } \textnormal{vertices} \P(X_e) + \P(X_e) \textnormal{ and edges } \P(X_e)
		$$
		$$
			s,t:\P(X_e) \to \P(X_e) + \P(X_e) \textrm{ are the coproduct injections}
		$$
		That is, $K(X)$ is the disjoint union of arrows indexed by $\P(X_e)$. The definition on morphisms is again as expected. Let
		\begin{equation*}
			\sigma, \tau: H \to K
		\end{equation*}
		be the natural transformations corresponding to $s$ and $t$: for every $M \subseteq X_e$, $\sigma_X(M)$ is the source of the arrow labelled by $M$ and $\tau_X(M)$ is its target.
		The coequalizer $L$ of $\sigma$ and $\tau$ in $[\Gra,\Gra]$ is obvious: it assigns to every graph $X$ the graph on $\P(X_e)$ consisting of loops:
		\begin{equation*}
			L(X) \textrm{ has vertices = edges = } \P(X_e) \textrm{ and } s= t
		\end{equation*}
		The functor $H$ generates a free monad, since in Construction \ref{\labelauthor:con:C-W} we have
		\begin{equation*}
			W_2 = X + H(X+HX) = H + HX = W_1.
		\end{equation*}
		Thus the construction converges in one step. The same is true about $K$.
		
		It remains to prove that $L$ does not generate a free monad. By Theorem \ref{\labelauthor:thm:T-B} it is sufficient to prove that $L$ does not have an initial algebra. Indeed, we prove that if
		\begin{equation*}
			a: LA \to A
		\end{equation*}
		is an initial algebra, then $\P$ has an initial algebra (compare Example \ref{\labelauthor:ex:E-P}). Let $m:A_0 \to A$ be the subgraph of $A$ whose vertices are precisely the loops of $A$ and whose edges are just all the loops. Then $LA = LA_0$ , and we obviously have a codomain restriction $a_0: LA_0 \to A_0$ of $a$. And $m: (A_0, a_0) \to (A,a)$ is a homomorphism of algebras for $L$. The unique homomorphism $h:(A,a) \to (A_0,a_0)$ thus yields an endomorphism $m \cdot h$ of the initial algebra; hence $m \cdot h = \id$. This proves $A = A_0$. That is, $A$ is the set $A_v$ of vertices endowed with all loops. But then $a: \P A_v \to A_v$ as an algebra for $\P$  is initial: given any algebra $b: \P B \to B$, form the graph $\bar{B}$ of all loops in $B$ and obtain an obvious structure $\bar{b}: L \bar{B} \to \bar{B}$ of an $L$-algebra. Then $\P$-algebra homomorphisms from $(A,a)$ to $(\bar{B}, \bar{b})$ are precisely the $L$-algebra homomorphisms from $(A_v,a)$ to $\bar{B}$. This is the desired contradiction.
	\end{example}

	\begin{example}
		The category Monad ($\Gra$) also fails to have cointersections of split epimorphisms. The argument is completely analogous: the following split epimorphisms
		$$
			\sigma_0 = [\sigma, \sigma,\tau, \id]: H+H+H+K \to K
		$$
		and
		$$
			\tau_0 = [\tau,\sigma,\tau, \id]: H+H+H+K \to K
		$$
		of $[\Gra,\Gra]$ have the cointersection as follows:
		\begin{equation*}
			\xymatrix{
				&
				\hat{H} = H+H+H+K \ar[ld]_{\sigma_0} \ar[rd]^{\tau_0}
				\\
				K \ar[rd]_{\rho}
				&&
				K \ar[ld]^{\rho}
				\\
				&
				L	
			}
		\end{equation*}
		Since $L$ does not generate a free monad, the split epimorphisms $\bar{\sigma}_0, \bar{\tau}_0: \mathds{F}_{\hat{H}} \to \mathds{F}_K$ do not have a cointersection in Monad ($\Gra$).
	\end{example}		

\section{Coproducts of Separated Monads}
	Ghani and Ustalu presented in \cite{\labelauthor:bib:gu} an interesting formula for coproducts of ideal monads, see Example \ref{\labelauthor:ex:example-42}(4), which was, in case of monads over $\Set$, generalized in \cite{\labelauthor:bib:ablm}. The present section is based on the ideas of the latter paper, extending the formula to separated monads over abstract categories. Separatedness means that the monad unit has a complement -- not over the given category $\A$ but over the category $\A_m$ of all objects and all monomorphisms.
	
	\begin{assumption}
		Thoughout this section $\A$ denotes a cocomplete category in which a coproduct of parallel monomorphisms is always monic.
		
		We denote by 
		\begin{equation*}
			\A_m
		\end{equation*}				
		the category of all objects and all monomorphisms of $\A$.
		
		Every monos-preserving endofunctor $F$ of $\A$ defines an endofunctor of $\A_m$ by restriction, we denote it by $F$ again.
		
		The coproduct $+$ of $\A$ is a monoid structure on $\A_m$ (not having the universal property of coproducts, of course).
	\end{assumption}
	
	\begin{example} \label{\labelauthor:ex:E-Set}
		\begin{enumerate}[(1)]
			\item The \emph{exception monad} $\bkM_E$ defined by $X \mapsto X+E$ has coproduct with all monads $\bkS$: the coproduct is given by $X \mapsto S(X+E)$.
			\item The terminal monad $\mathds{1}$ given by $X \mapsto 1$, also has all coproducts, the result is always $\mathds{1}$.
			\item For monads over $\Set$ there are essentially no other monads having a coproduct with every monad. More precisely, let $\bkM^0_E$ be the modification of $\bkM_E$ with $\emptyset \mapsto \emptyset$ and $X \mapsto E$ for all $X \neq \emptyset$. Analogously, let $\mathds{1}^0$ be given by $\emptyset \mapsto \emptyset$ and $X \mapsto 1$ for all $X \neq \emptyset$. It is easy to see that Monad ($\Set$) has all coproducts with $\bkM^0_E$ or with $\mathds{1}^0$.
		\end{enumerate}
	\end{example}
	
	\begin{definition}
		We call a monad over $\Set$ \emph{trivial} if it is isomorphic to $\mathds{M}_E$, $\mathds{M}^0_E$, $\mathds{1}$, or  $\mathds{1}^0$. These are precisely the monads corresponding to varieties of alegebras with no operation of arity at least 1.
	\end{definition}		
	
	\begin{theorem}[See \cite{\labelauthor:bib:ablm}] \label{\labelauthor:thm:T-44}
		A monad over $\Set$ has coproducts with all monads iff it is trivial.
		
		Moreover, all monads over $\Set$ except $\mathds{1}$ and $\mathds{1}^0$ are \emph{consistent}, i.\,e., the components of the monad unit are monic.
	\end{theorem}
	
	\subsection{The category of multi-algebras}
		Given a discrete diagram $\D$ of monads $\bkT_i\,(i \in I)$ the category $\C_D$ of Remark \ref{\labelauthor:rem:R-Kelly} has as objects \emph{multi-algebras}
		\begin{equation*}
			(A,(a_i)_{i\in I}) \textnormal{ where } a_i: T_i A \to A \textnormal{ lies in }\A^{\bkT_i}
		\end{equation*}
		and morphisms are those maps in $\A$ that are homomorphisms for each of $\bkT_i$ simultaneously. A coproduct of the monads $\bkT_i$ exists in Monad ($\A$) whenever every object of $\A$ generates a free multi-algebra.
	
%		We are going to prove that every collection of separated monads (see the next definition) whose underlying functors have arbitrarily large joint fixpoints, has a coproduct.
	
	\begin{definition}
		A monad $(S, \mu, \eta)$ is called \emph{separated} if its unit has a complement in the following sense:
		\begin{enumerate}
			\item[(i)] $S$ preserves monomorphisms
		\end{enumerate}
		and
		\begin{enumerate}
			\item[(ii)] there exists an endofunctor $\bar{S}$ of $\A_m$ such that
			\begin{equation*}
				S = \Id + \bar{S} 
			\end{equation*}
			with the unit $\eta$ as the left-hand injection.
		\end{enumerate}
	\end{definition}
	
	\begin{examples} \label{\labelauthor:ex:example-42}
		\begin{enumerate}[(1)]
			\item The exception monad $\bkM_E$ is separated: here $\bar{M}_E$ is the constant functor of value $E$.
			\item Every free monad $\bkF_H$ which preserves monomorphisms is separated. (In particular, if $\A$ has stable monomorphisms, all free monads on monos-preserving functor are separated.) Here $\bar{\bkF}_H = H \cdot \bkF_H$: use Remarks \ref{\labelauthor:rem:r-mono} and \ref{\labelauthor:rem:r-plus}.
			\item All consistent monads on $\Set$ (i.\,e., all except $\mathds{1}$ and $\mathds{1}^0$) are separated. See \cite{\labelauthor:bib:ablm}, Proposition IV.5.
			\item \emph{Ideal monads} of Elgot \cite{\labelauthor:bib:e} are separated if they preserve monomorphisms. Recall that an ideal monad $\bkS = (S, \mu, \eta)$ is one for which an endofunctor $\bar{S}$ of $\A$ exists such that (i) $S = \Id +\bar{S}$ in $[\A,\A]$ with the left-hand injection $\eta$ and (ii) $\mu$ restricts to a natural transformation $\bar{\mu}: \bar{S}S \to \bar{S}$.
			\item In particular, the free completely iterative monad $\bkS$ on an endofunctor $H$ given by the greatest fixpoint
			\begin{equation*}
				SA = \nu X \cdot (A + HX)
			\end{equation*}
			is separated, with $\bar{S} = H \cdot S$, whenever it preserves monomorphisms, see \cite{\labelauthor:bib:AAMV}.
		\end{enumerate}
	\end{examples}
	
	\begin{notation}
		Let $\bkS_i\,(i \in I)$ be separated monads. For every object $A$ of $\A$ define an endofunctor $H_A$ of $\A^I_m$ as follows:
		\begin{equation*}
			H_A(X_i)_{i \in I} = (\bar{S}_i Y_i)_{i \in I} \textnormal{ where } Y_i = A+ \coprod_{j \in I, j \neq i} X_j
		\end{equation*}
		If $H_A$ has an initial algebra, we denote its components by $S^\ast_i A$:
		\begin{equation*}
			\mu H_A = (S^\ast_i A)_{i \in I}
		\end{equation*}
	\end{notation}
	
	\begin{remark}
		Let $(X_i)$ be a fixed point of $H_A$:
		\begin{equation*}
			X_i \simeq \bar{S}_i Y_i \textnormal{ for all } i \in I
		\end{equation*}
		Then the coproduct $A+ \coprod_{i \in I} X_i$ carries a canonical structure of a multi-algebra: the algebra structure for $\bkS_i$ is the free algebra on $Y_i$. Indeed, the usual free algebra is ($S_iY_i, \mu^i_{Y_i})$. And the above coproduct is isomorphic to $S_iY_i:$.
		\begin{align*}
			A + \coprod_{i \in I} X_i &\cong Y_i + X_i\\
			&\simeq Y_i + \bar{S}_i Y_i\\
			&= S_iY_i
		\end{align*}
		In particular: if the initial algebra $\mu H_A = (S^\ast_i)_{i \in I}$ exists, then the coproduct
		\begin{equation*}
			A + \coprod_{i \in I} S_i^\ast A
		\end{equation*}
		is a multi-algebra. We prove that it is free on $A$ w.r.t. the right-hand coproduct injection $\inl : A \to A + \coprod_{i \in I} S^\ast_i A$:
	\end{remark}
	
	\begin{theorem}
		A coproduct of separated monads $\bkS_i \, (i \in I)$ exists whenever the initial algebra $\mu H_A = (S^\ast_i A)$ exists for every object $A$. It is defined by
		\begin{equation*}
			A \mapsto A + \coprod_{i \in I} S^\ast_i A
		\end{equation*}
	\end{theorem}
	
	\begin{remark}
		The monad unit $\eta_A$ is the right-hand coproduct injection. The multiplication follows from $A + \coprod S^\ast_i A$ being the free multi-algebra on $A$.
	\end{remark}
	
	\begin{proof}
		Let $\bkS_i = (S_i, \mu^i, \eta^i)$ be the given monads. Following Remark \ref{\labelauthor:rem:R-Kelly} all we need proving is that the multi-algebra $\bar{A}= A + \coprod_{i \in I} S^\ast_i A$ is free.
		\begin{enumerate}[(1)]
			\item Let us describe its algebra structure explicitly for every $\mathds{S}_i$. The initial-algebra structure of $\mu H_A$ is given by isomorphisms
			\begin{equation*}
				\varphi_i: \bar{S}_i Y_i \to X_i
			\end{equation*}
			where $X_i = S_i^\ast A$ and $Y_i = A + \coprod_{ j \neq i} X_j$. This defines isomorphisms
			\begin{equation} \label{\labelauthor:eq:eqn41}
				\bar{\varphi}_i \equiv \bar{A} = Y_i + X_i \xrightarrow{Y_i + \varphi_i^{-1}} Y_i + \bar{S}_i Y_i = S_i Y_i
			\end{equation}
			And the algebra structure $\sigma_i$ of $\bar{A}$ for $\bkS_i$ is transported by this isomorphism from the free-algebra structure $\mu^i_{Y_i}$:
			\begin{equation*}
				\xymatrix{
					S_i \bar{A} \ar[r]^{\sigma_i} \ar[d]_{S_i \tilde{\varphi}_i}
					&
					\bar{A}
					\\
					S_iS_iY_i \ar[r]_{\mu^i_{Y_i}}
					&
					S_iY_i \ar[u]_{\bar{\varphi}_i^{-1}}
				}
			\end{equation*}
			\item For every multi-algebra
			\begin{equation*}
				\beta_i : S_i B \to B \qquad (i \in I)
			\end{equation*}
			and every morphism $f:A \to B$ we prove that a unique multi-algebra homomorphism
			\begin{equation*}
				\bar{f}: \bar{A} \to B \textnormal{ with } f = \bar{f} \cdot \inl
			\end{equation*}
			exists. The object $\triangledown B = (B,B,B \ldots)$ is an algebra for $H_A$ w.r.t. $(b_i)_{i \in I}:H_A(\triangledown B) \to \triangledown B$ given as follows:
			\begin{equation*}
				b_i \equiv 
				\xymatrix@1{
					\bar{S}_i(A+ \coprod_{I - \{i\}} B) \ar[r]^-{\bar{S}_i[f,\triangledown]} 
					&
					\bar{S}_iB~ \ar@{^{(}->}[r]
					&
					S_iB \ar[r]^{\beta_i}
					&
					B
				}
			\end{equation*}
			The middle subobject is the right-hand coproduct injection of $S_i B = B + \bar{S}_i B$. We have a unique homomorphism from the initial algebra $\mu H_A$:
			\begin{equation*}
				(h_i)_{i \in I}: (X_i)_{i \in I} \to \triangledown B
			\end{equation*}
			which means that the square
			\begin{equation} \label{\labelauthor:eq:eqn42}
				\vcenter{
					\xymatrix{
						\bar{S}_iY_i = \bar{S}_i (A + \coprod_{j \neq i} X_j) \ar[r]^-{\varphi_i} \ar[d]_{\bar{S}_i(A + \coprod_{j \neq i} h_j)}
						&
						X_i \ar[d]^{h_i}
						\\
						\bar{S}_i (A + \coprod_{j \neq i} B) \ar[r]_-{b_i}
						&
						B
					}
				}
			\end{equation}
			commutes for every $i$. Put
			\begin{equation*}
				\bar{h}_i = [h_j]_{j \neq i}: \coprod_{j \in I, j \neq i} X_j \to B
			\end{equation*}
			Then (\ref{\labelauthor:eq:eqn42}) is equivalent to the commutativity of the following square:
			\begin{equation} \label{\labelauthor:eq:eqn43}
				\vcenter{
					\xymatrix{
						\bar{S}_iY_i \ar[rr]^{\varphi_i} \ar[d]_{\bar{S}_i[f,\bar{h}_i]}
						&&
						X_i \ar[d]^{h_i}
						\\
						\bar{S}_i B \ar@{^{(}->}[r]
						&
						S_i B \ar[r]_{\beta_i}
						&
						B
					}
				}
			\end{equation}
			We are going to prove that the desired extension of $f$ is
			\begin{equation*}
				\bar{f} = [f ,\bar{h}]:A + \coprod_{j \in I} X_j \to B \textnormal{ where } \bar{h} = [h_j]_{j \in I}.
			\end{equation*}
			That is, we first need to prove that $\bar{f}$ is a homomorphism for $H_A$. Thus for every $i \in I$ we must prove that the following diagram commutes:
			\begin{equation}
				\vcenter{
					\xymatrix{
						S_i \bar{A} \ar[r]^{S_i \tilde{\varphi}_i} \ar[dd]_{S_i \bar{f}}
						&
						S_iS_iY_i \ar[r]^{\mu^i_Y} \ar[d]_{S_iS_i \bar{f}}
						&
						S_iY_i \ar[r]^-{\tilde{\varphi}^{-1}_i} \ar[d]^{S_i \bar{f}}
						&
						\bar{A} = A + \coprod_{j \neq i} X_j \ar[dd]^{\bar{f}}
						\\
						&
						S_iS_iB \ar[r]_{\mu^i_B} \ar[dl]_{S_i\beta_i}
						&
						S_i B \ar[rd]^{\beta_i}
						\\
						S_i B \ar[rrr]_{\beta_i}
						&&&
						B
					}
				}
			\end{equation}
			(The upper line is the algebra structure $\sigma_i$ of $\bar{A}$.) The middle square is the naturality of $\mu^i$, the lower-one is a monad-algebra axiom for $(B,\beta_i)$. We only need to prove that the right-hand square commutes: the left-hand one is its image under $S_i$. Using $S_iY_i = Y_i + \bar{S}_i Y_i$ we get the following presentation of the right-hand square, recalling (\ref{\labelauthor:eq:eqn41}):
			\begin{equation*}
				\xymatrix@+15pt{
					S_iY_i = Y_i + \bar{S}_iY_i \ar[r]^-{Y_i + \varphi_i} \ar[d]_{\bar{S}_i[f, [\bar{h}_i]] + [f, [\bar{h}_i]]}
					&
					A + \coprod_{j \neq i} X_j \ar[d]^{[f, [\bar{h}_i]]}
					\\
					B + \bar{S}_i B \ar[r]_{\beta_i}
					&
					B
				}
			\end{equation*}
			The left-hand component with domain $Y_i$ clearly commutes: recall that $\eta^i_B = \inl : B \to S_i B$, thus $\beta_i \cdot \inl = \id$ due to the monad axioms for $(B, \beta_i)$. The right-hand component forms the square (\ref{\labelauthor:eq:eqn43}).
			\item To prove uniqueness, let $\bar{f}: \bar{A} \to B$ be a multi-algebra homomorphism with $\bar{f} \cdot \inl = f$. Define $h_i: X_i \to B$ to be the $i$-th component of $\bar{f}$, thus, $\bar{f} = [f, [\bar{h}_i]]$. It is only needed to prove that the squares (\ref{\labelauthor:eq:eqn42}) commute: then $h_i$'s are determined uniquely, since $(X_i)$ is the initial algebra of $H_A$. Since $\bar{f}$ is a multi-algebra homomorphism, (\ref{\labelauthor:eq:eqn43}) commutes. This clearly implies that (\ref{\labelauthor:eq:eqn42}) does.
		\end{enumerate}
	\end{proof}
	
	\begin{theorem}[See \cite{\labelauthor:bib:ablm}] \label{\labelauthor:thm:T-411}
		For monads over $\Set$ the above sufficient condition for coproducts is essentially necessary: a coproduct of separated (= consistent) monads exists iff
		\begin{enumerate}[(a)]
			\item for every set $A$ the initial algebra of $H_A$ exists
		\end{enumerate}
		or
		\begin{enumerate}[(b)]
			\item all but one of the monads is trivial (i.\,e. isomorphic to $\mathds{M}_E$ or $\mathds{M}^0_E$).
		\end{enumerate}
	\end{theorem}		
	
	\begin{corollary}
		Let $\A$ have stable monomorphisms. Every collection of separated monads with arbitrarily large joint pre-fixpoints has a coproduct in Monad ($\A$).
	\end{corollary}
		
	Indeed, assuming $\bkS_i, i \in I$, have arbitrarily large joint pre-fixpoints, we prove that the endofunctor $H_A$ has an initial algebra. By Corollary \ref{\labelauthor:cor:C-free} we only need to find, for every object $X = (X_i)$ of $\A^I_m$, a prefixed point $Z$ of $H_A$ with $ Z \simeq Z + X$.
		
	The functor $S = \coprod_I \coprod_{i \in I} S_i$ has arbitrarily large pre-fixpoints: given an object $Y$ of $\A$, let $V$ be a joint pre-fixpoint of all $S_i$ with $Y+V \simeq V \simeq \coprod_{I+I} V$, then $V$ is a pre-fixpoint of $S$ due to
	\begin{equation*}
		SV = \coprod_I \coprod_{i \in I} S_i V \rightarrowtail \coprod_I \coprod_I V \simeq V.
	\end{equation*}
	By Corollary \ref{\labelauthor:cor:C-free}, $S$ has a free algebra on
	\begin{equation*}
		Y = A + \coprod_{\mathds{N}} \coprod_{i \in I} X_i
	\end{equation*}
	(for the above object $X$ of $\A^I$). Put $Y^\ast = F_SY$. Remark \ref{\labelauthor:rem:r-plus} yields
	\begin{equation*}
		Y^\ast = SY^\ast + Y = SY^\ast + A + \coprod_{\mathds{N}} \coprod_{i \in I} X_i.
	\end{equation*}
	The desired object $Z$ of $\A^I_m$ is $Z = (Y,Y,Y, \ldots)$. Obviously $Y_i \cong Y_i + X_i$, thus, $Z \simeq Z + X$. And $H_AZ = (A + \coprod_{j \neq i} \bar{S}_j Y^\ast)_{i \in I}$ is a subobject of $Z$ due to the following monomorphism:
	\begin{equation*}
		A + \coprod_{j \neq i} \bar{S}_j Y^\ast \rightarrowtail Y + \coprod_I SY^\ast \simeq Y + SY^\ast \simeq Y^\ast.
	\end{equation*}
	
	\begin{corollary}
		Let $\A$ have stable monomorphisms. A coproduct of accessible separated monads $\bkS$ and $\bkT$ is given by
		\begin{equation*}
			 A \mapsto A + \colim X_k + \colim Y_k
		\end{equation*}
		for the transfinite chains
		\begin{equation*}
			X_k: 0 \to \bar{S}A \to \bar{S}(A + \bar{T}A) \to \ldots
		\end{equation*}
		and
		\begin{equation*}
			Y_k: 0 \to \bar{T}A \to \bar{T}(A + \bar{S}A) \to \ldots
		\end{equation*}
		More precisely, there chains are defined by the mutual recursion
		\begin{equation} \label{\labelauthor:eqn:Rec2}
			X_{k+1} = \bar{S}(A+Y)_k \textnormal{ and } Y_{k+1} = \bar{T}(A+X_k)
		\end{equation}
		on isolated steps, and by colimits on limit steps.
	\end{corollary}	
	To see this, let $\lambda$ be an infinite cardinal such that $S$ and $T$ preserve $\lambda$-filtered colimits. Then $\bar{S}$ and $\bar{T}$ also preserve $\lambda$-filtered colimits. (Indeed, given a $\lambda$-filtered colimit $b_j: B_j \to B, j \in J$, we know that $Sb_j = b_j + \bar{S}b_j$ is also a colimit cocone. For every cocone $c_j : \bar{S}B_j \to C$ consider the cocone $b_j + c_j: SB_j \to B+C$. Since this factorizes uniquely through $Sb_j$, it follows that $c_j$ factorizes uniquely though $\bar{S}b_j$. Thus $\bar{S}$ preserves $\lambda$-filtered colimits, analogously $\bar{T}$). Consequently, the functor $H_A(V,W) = (\bar{S} (W+A), \bar{T}(V+A))$ preserves $\lambda$-filtered colimits. This implies, as proved in $[A]$, that $\mu H_A$ is the colimit of the $\lambda$-chain $(X_i, Y_i)$ which is the free-algebra chain $H_A$ and the initial object $X$ of $\A^I_m$, see Construction \ref{\labelauthor:con:C-W}. The recursion $W_{k+1} = H_A W_k$ is precisely (\ref{\labelauthor:eqn:Rec2}) above.
	
	\begin{remark}
		More generally, a coproduct of accessible separated monads $\bkS_i\, (i \in I)$ is given by
		\begin{equation*}
			A \mapsto A + \coprod_{i \in I} X^i_k
		\end{equation*}
		for the transfinite $(X^i_k)_{k \in \Ord}$ chains given on isolated steps by
		\begin{equation*}
			X^i_{k+1} = \bar{S}_i (A + \coprod_{j \neq i} X^j_k)
		\end{equation*}
		and on limit steps by colimits.
	\end{remark}
	
	\begin{notation}
		For a separated monad $\bkS$ define endofunctors $\bar{S}_A$ of $\A_m$ by
		\begin{equation*}
			\bar{S}_A X = \bar{S} (A+X).
		\end{equation*}
		Thus, the above formula simplifies to $X^i_{k+1} = ( \bar{S}_i )_A \coprod_{j \neq i} X^j_k$. For two monads we also have a more compact formula:
	\end{notation}
	
	\begin{corollary}
		Let $\A$ have stable monomorphisms.
		The coproduct of a pair $\bkS, \bkT$ of separated monads with arbitrarily large joint pre-fixpoints is given by
		\begin{equation*}
			A \mapsto A + \mu \bar{S}_A \bar{T}_A + \mu \bar{T}_A \bar{S}_A.
		\end{equation*}
	\end{corollary}	
	Indeed, the coproduct is given by $A + S^\ast A + T^\ast A$, so all we need proving is that the endofunctor $H_A$ has the initial algebra carried by $(\mu \bar{S}_A \bar{T}_A, \mu \bar{T}_A \bar{S}_A)$. We prove a more general statement:

	\begin{lemma}
		Given endofunctors $F$ and $G$ of $\A$ define an endofunctor $H$ of $\A^2$ by $H(V,W) = (FW, GV)$. If $(X,Y)$ is an initial algebra of $H$, then $X = \mu FG$ and $Y = \mu GF$.
	\end{lemma}
	
	\begin{proof}
		Let the algebra structure of $\mu H = (X,Y)$ be given by
		\begin{equation*}
			\xymatrix@1{x:FY \ar[r]^-{\sim} & X} \textnormal{ and } \xymatrix@1{y:GY \ar[r]^-{\sim} & X}.
		\end{equation*}
		Then we prove that $GF$ has the initial algebra
		\begin{equation*}
			\xymatrix@1{GFY \ar[r]^{Gx} & GX \ar[r]^{y} & X,}
		\end{equation*}
		by symmetry $\mu FG = X$.
		
		For every algebra $\beta:GFB \to B$ of $GF$ form the algebra for $H$  on $(FB,B)$ with the following structure
		\begin{equation*}
			\id: FB \to FB \textnormal{ and } \beta: GFB \to B.
		\end{equation*}
		Given the unique homomorphism of $H$-algebras
		\begin{equation*}
			(a,b): (X,Y) \to (FB,B)
		\end{equation*}
		it is easy to verify that $b: (X,y . Gx) \to (B,\beta)$ is a homomorphism for $GF$. Conversely, if $b: (X, y . Gx) \to (B, \beta)$ is a homomorphism for $GF$, then put $a = Fb . x^{-1}: X \to FB$. Then $(a,b): (X,Y) \to (FB,B)$ is a homomorphism for $H$. Thus, $b$ is the unique homomorphism for $GF$, proving $\mu GH = X$. 
	\end{proof}

%% ----------------------- REFERENCES -------------------------------%%

\end{document}